\documentclass[12pt]{amsart}
\usepackage{latexsym,amssymb,amsmath,mathtools,hyperref}
\usepackage{xcolor,enumerate,tikz,comment}
\usetikzlibrary{calc,positioning}

\numberwithin{equation}{section}
\newtheorem {theorem}[equation]{Theorem}
\newtheorem*{theorem*}{Theorem}
\newtheorem {lemma}[equation]{Lemma}
\newtheorem {corollary} [equation]      {Corollary}
\newtheorem*{corollary*}{Corollary}
\newtheorem {proposition}[equation]     {Proposition}

\theoremstyle{definition}
\newtheorem {definition}[equation]{Definition}

\theoremstyle{remark}
\newtheorem {remark}[equation]{Remark}
\newtheorem {example}[equation]{Example}

\newcommand{\init}{\mathrm{in}}

\newcommand{\bigslant}[2]{{\raisebox{.2em}{$#1$}\left/\raisebox{-.2em}{$#2$}\right.}}

\author[S. Da Silva]{Sergio Da Silva}
\author[A. Stewart]{Aniya Stewart}

\address[Sergio Da Silva]{
    Dept.\ of Mathematics and Economics, 
    Virginia State University, 
    1 Hayden Drive,
    Petersburg, Virginia 23806, USA}
\email{sdasilva@vsu.edu, smd322@cornell.edu}

\address[Aniya Stewart]{
 Dept.\ of Mathematics and Economics, 
    Virginia State University, 
    1 Hayden Drive,
    Petersburg, Virginia 23806, USA}
\email{aste9039@students.vsu.edu, stewart.aniya@gmail.com}

\title[Quantum-Resistant Cryptography via Universal Gr\"obner Bases]{Quantum-Resistant Cryptography via Universal Gr\"obner Bases}

\date{\today}

\begin{document}

\begin{abstract}

In this article, we explore the use of universal Gr\"obner bases in public-key cryptography by proposing a key establishment protocol that is resistant to quantum attacks. By utilizing a universal Gr\"obner basis $\mathcal{U}_I$ of a polynomial ideal $I$ as a private key, this protocol leverages the computational disparity between generating the universal Gr\"obner basis needed for decryption compared with the single Gr\"obner basis used for encryption. The security of the system lies in the difficulty of directly computing the Gr\"obner fan of $I$ required to construct  $\mathcal{U}_I$. We provide an analysis of the security of the protocol and the complexity of its various parameters. Additionally, we provide efficient ways to recursively generate $\mathcal{U}_I$ for toric ideals of graphs with techniques which are also of independent interest to the study of these ideals. 

\end{abstract}

\maketitle

\section{Introduction} \label{sec: intro}

Cryptographic systems often rely on the computational difficulty of solving particular mathematical problems. Quantum computing is a rapidly growing industry \cite{SP} with reports of capable quantum systems being available by as soon as 2030, making post-quantum cryptography especially relevant while also threatening the security of traditional cryptographic methods \cite{WHHWJDRM}. For example, the commonly used RSA cryptosystem, which relies on the difficulty of factoring the product of two (secret) large prime numbers \cite{LLMA}, would no longer remain secure using Shor’s algorithm on a quantum computer \cite{WWJDRM}. As the vulnerabilities facing  cryptographic systems becomes a reality, it is necessary to explore other approaches and techniques that might prove more useful in resisting quantum attacks.

One promising area of exploration is with primitives that utilize algebraic or combinatorial constructions, especially in the context of lattice-based cryptography. Many algebraic constructions utilize computational aspects of ideals in polynomial rings \cite{CCT,CMT}, making  Gr\"obner bases a natural component in their implementation \cite{E}. Gr\"obner bases are specific generators of a polynomial ideal that allow many algebro-geometric properties to be computed efficiently from an associated monomial ideal. Past attempts to use Gr\"obner bases in public-key cryptography have failed, such as with Barkee cryptosystems \cite{BCEMR, BCMV}. The main obstacle to these approaches is that a single Gr\"obner basis is generally too easy to compute to realistically be used to secure a system. A universal Gr\"obner basis on the other hand is difficult to compute, and involves the  computation of high-dimensional lattice structures like the Gr\"obner fan and state polytope of a polynomial ideal. Our approach is to have one party use a universal Gr\"obner basis to produce a private list of keys while also having a public mechanism for another party to generate one key from that list.

Let $A$ and $B$ be two parties who have not previously communicated to share a common encryption key. To establish the protocol $\mathcal{P}$, Party $A$ starts with an ideal $I \subset \mathbb{K}[x_1,\ldots, x_n]$, a universal Gr\"obner basis $\mathcal{U}_I$ of $I$, and some generating set $\mathcal{R}_I$ of $I$. Then $\mathcal{R}_{I}$, $\mathbb{K}$, and the number of variables is made public, together with information about the encryption scheme needed to create the ciphertext, including two hash functions $\eta$ and $\tau$. The protocol $\mathcal{P}$ can be defined without the use of $\tau$, in which case we set $\tau=\emptyset$.

Now Party $B$ can send an encrypted message to $A$ using the information publicly provided by $A$. Party $B$ starts by choosing a random monomial order $<_B$ of $\mathbb{K}[x_1,\ldots, x_n]$ and then computes the initial ideal $\init_{<_B}\langle \mathcal{R}_I\rangle$. Both $<_B$ and the initial ideal are kept private. By Dickson's lemma, there is a unique minimal generating set of $\init_{<_B}\langle \mathcal{R}_I\rangle$, and based on a public hash function $\eta$ provided by $A$, Party $B$ converts this set into a binary sequence which will serve as the encryption key $K_B$. Using the predetermined encryption scheme $\mathcal{E}$ provided, $B$ encrypts the message into a ciphertext. If $\tau=\emptyset$, then only this ciphertext is sent back to $A$. Otherwise, $\tau(K_B)$ is also sent back to $A$. The protocol can be described using the following schematic:

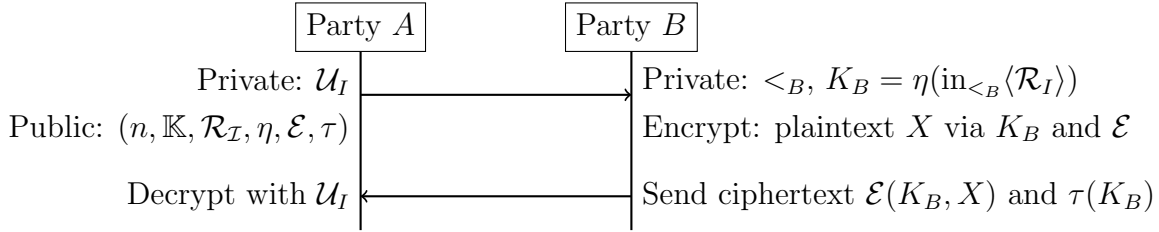
\begin{center}
\begin{tikzpicture}[scale=0.9]
  
  % Alice
  \node[draw] (Alice) at (-2,0) {Party $A$}; 
  \draw[thick] (Alice) -- ++(0, -3);
    
  % Calculations of Alice
  \node[draw=none,fill=none,anchor=east] (node) at ($(Alice) + (0,-0.8)$) {Private: $\mathcal{U}_I$}; 
    \node[draw=none,fill=none,anchor=east] (node) at ($(Alice) + (0,-1.5)$) {Public: $(n,\mathbb{K},\mathcal{R_I},\eta,\mathcal{E},\tau)$};
    \node[draw=none,fill=none,anchor=east] (node) at ($(Alice) + (0,-2.5)$) {Decrypt with $\mathcal{U}_I$};

  % Bob
  \node[draw] (Bob) at (2,0) {Party $B$}; 
  \draw[thick] (Bob) -- ++(0, -3);
   
  % Calculations of Bob
  \node[draw=none,fill=none,anchor=west] (path) at ($(Bob) + (0,-0.8)$) {Private: $<_B$, $K_B=\eta(\init_{<_B}\langle\mathcal{R}_{I}\rangle)$};
   \node[draw=none,fill=none,anchor=west] (path) at ($(Bob) + (0,-1.5)$) {Encrypt: plaintext $X$ via $K_B$ and $\mathcal{E}$};
    \node[draw=none,fill=none,anchor=west] (path) at ($(Bob) + (0,-2.5)$) {Send ciphertext $\mathcal{E}(K_B,X)$ and $\tau(K_B)$};

  % Messages
  \draw[->,thick] ($(Alice)+(0,-1)$) -- ($(Bob)+(0,-1)$) node [pos=0.5,above,font=\footnotesize] {};
    \draw[->,thick] ($(Bob)+(0,-2.5)$) -- ($(Alice)+(0,-2.5)$) node [pos=0.4,below,font=\footnotesize] {};
\end{tikzpicture}
\end{center}

First suppose that $\tau=\emptyset$. If an attacker were to intercept the ciphertext, they wouldn't know what parameters $B$ chose to produce the ciphertext, so a brute-force attack would be intractable. For instance, with RSA public-key cryptography, an attacker at least knows the product $m=pq$ (which is public) and could try to find a brute-force factorization of $m$. With our setup, if an RSA encryption algorithm $\mathcal{E}$ were used to create a ciphertext, $m$ would remain hidden, so an attacker wouldn't even know what number to attempt to factor.

This added security comes at a cost. Since Party $A$ will also not know which key Party $B$ chose, the only option is to try all possible keys to decrypt the message. However, $A$ has the universal Gr\"obner basis $\mathcal{U}_I$, and therefore possesses the list of all possible initial ideals of $I$, and hence has the list of all possible encryption keys that $B$ could have generated. The list that Party $A$ needs to exhaustively search to decrypt the message is reasonable compared to the intractable list that an attacker would need to try. The system also relies on $\mathcal{U}_I$ being extremely difficult to compute directly without prior knowledge of any symmetries used to construct $I$.

On the other hand, if $\tau\neq \emptyset$, then $A$ could decrypt the message without iterating through the list of keys since the image of each key under $\tau$ would already be known and could be compared with the value $\tau(K_B)$ provided by $B$. However, this reduces the overall security of the system since an attacker intercepting the message would also have knowledge of $\tau(K_B)$, and could try attacking $\tau$ instead of computing $\mathcal{U}_I$ directly. We summarize the analysis of the protocol $\mathcal{P}$ presented in this article in the following theorem.

\begin{theorem*}
Let $\mathcal{P}= (n,\mathbb{K}, \mathcal{U}_I, \mathcal{R}_I, <_B, \eta, \mathcal{E},\tau)$ be a protocol as in Definition \ref{def: par}. Then the following statements about the complexity and security of the cryptosystem hold:

\begin{itemize}
    \item Party $A$ can (privately) construct $\mathcal{U}_I$ in polynomial time by using Theorem \ref{thm: construct U_I}. Given $\mathcal{U}_I$, the set $\mathcal{R}_I$ can be computed in linear time by Theorem \ref{thm: trim ideal}.
    \item Party $B$ can send a message to $A$ by computing a single Gr\"obner basis for $I$, the complexity of which is summarized in Theorem \ref{thm: GB complexity}.
    \item The amount of time that $A$ requires to decrypt the ciphertext when $\tau=\emptyset$ is summarized in Proposition \ref{prop: decrypt}.
    \item An attacker without any trapdoor knowledge about how $I$ was constructed would need to compute the Gr\"obner fan of $I$ directly, which is NP-hard with complexity described in Theorem \ref{thm: Gfan}.
    \item By Theorem \ref{thm: secure}, if lattice-based primitives secured by the Shortest Vector Problem (SVP) are quantum-resistant, then so is the  protocol $\mathcal{P}$ when $\tau=\emptyset$.
\end{itemize}
\end{theorem*}

A more detailed discussion of the steps involved to initialize the protocol $\mathcal{P}$ is presented in Section \ref{sec: protocol}. We then consider practical complexity and security issues in Section \ref{sec: security} and subsequently describe how to efficiently construct $\mathcal{U}_I$ in Section \ref{sec: initial} using the toric ideal $I_G$ of a graph $G$. This ideal can be generated by binomials corresponding to primitive closed even walks of $G$, which incidentally also define a universal Gr\"obner basis of $I_G$. Using four graph operations, one can recursively generate large graphs for which $\mathcal{U}_I$ is computable and contains enough primitive closed even walks to ensure the security of $\mathcal{P}$.

It is worth noting that the graph constructions presented in Section \ref{sec: initial}, together with their effect on $\mathcal{U}_I$, are of independent importance to combinatorial algebraists \cite{BC,N,TT,V}. For instance, previous research in this area has revealed connections with algebraic statistics \cite{DLST}, commutative-algebraic techniques  \cite{CDSRVT,GHKKPVT}, and network complexity \cite{DSNR}. In the final section of this article, we will discuss weaknesses and practical concerns regarding the protocol $\mathcal{P}$. We also identify areas for potential future research.

\noindent
{\bf Acknowledgments.} 
We thank Sarah Arpin for many helpful conversations and references on post-quantum cryptography. Da Silva's research is supported by NSF LEAPS-MPS Grant 2532757.

\section{Preliminaries} \label{sec: prelim}

In this paper, $\mathbb{K}$ will denote any field. In this section, we will provide a very  brief overview of Gr\"obner bases, Gr\"obner fans and state polytopes. There is a considerable amount of theory involved with these topics, so we provide only what is necessary to understand the subsequent sections, and refer the reader to \cite{CLO,E,S} for further details.

\subsection{Gr\"obner bases} Gr\"obner theory provides a way to associate a monomial ideal $\init_<(I)$ to an ideal $I\subset R=\mathbb{K}[x_1,\ldots, x_n]$. This is done in such a way that many algebro-geometric properties of $I$ can be determined from $\init_<(I)$, especially since  monomial ideals are generally easier to study. A monomial order $<$ on $R$ is a total order on the monic monomials of $R$ such that $u\leq v \Rightarrow wu \leq wv$ and $1\leq w$ for any monomial $w\in R$. Given $\alpha = (\alpha_1,\alpha_2,\ldots,\alpha_n) \in \mathbb{Z}_{\geq 0}^{n}$, we will use the notation \[\mathbf{x}^\alpha:= x_1^{\alpha_1}\cdots x_n^{\alpha_n}.\]

\noindent With a monomial order $<$ on $R$ and a polynomial $f \in R$, we can order all terms of $f$ and define the initial term of $f$ as the monomial $\init_<(f)=c\mathbf{x}^\alpha, c\in\mathbb{K},$ which is greatest term of $f$ with respect to $<$.

\begin{definition} 
Let $<$ be a monomial order on $R$ and let $I \subseteq R$ be an ideal. The \emph{initial ideal} of $I$, denoted by $\init_<(I)$, is the monomial ideal in $R$ defined by
\begin{align*}
\init_<(I):=\langle \init_<(f)| f\in I \rangle.
\end{align*}
\end{definition}

Unfortunately, the initial terms of a generating set of $I$ do not generally  constitute a generating set of $\init_<(I)$. This leads us to the definition of a Gr\"obner basis.

\begin{definition}
Given an ideal $I \subseteq R$, a set $\mathcal{G}=\{g_1, \ldots, g_t\} \subset I$ is a \emph{Gr\"obner basis} for $I$ if $I= \langle g_1,\ldots,g_t \rangle \text{ and } \init_<(I)=\langle \init(g_1), \ldots,\init(g_t) \rangle.$ A \emph{universal Gr\"obner basis} $\mathcal{U}_I$ for an ideal $I$ is a generating set for $I$ which is a Gr\"obner basis for $I$ with respect to \emph{any} monomial order on $R$.
\end{definition}

\begin{example}
    
Let $I = \langle ag-bf, ce-dg \rangle \subset \mathbb{K}[a,b,c,d,e,f,g]$ and set $\mathcal{G} = \{ag-bf, ce-dg \}$. Given the lexicographic monomial  ordering $<_1$ defined by $a>b>c>d>e>f>g$ we can show (using \texttt{Macaulay2} for example) that \[\init_{<_1}(I) = \langle ag, ce \rangle = \langle\init_{<_1}(ag-bf), \init_{<_1}(ce-dg)\rangle,\] which implies that $\mathcal{G}$ is a Gr\"obner basis for $I$ with respect to $<_1$. 

If however we used the lexicographic monomial ordering $<_2$ defined by $d >a > b>c>e>f>g$, we would get $\init_{<_2}(I) = \langle ag,bdf, dg  \rangle$. Therefore, for $<_2$, $\mathcal{G}$ is not a Gr\"obner basis for $I$. To extend $\mathcal{G}$ to a Gr\"obner basis of $I$ with respect to $<_2$, we would need to include the polynomial
\begin{align*}
S(dg-ce,ag-bf) &= \frac{\init_{<_2}(ag-bf)}{\mathrm{gcd}(dg,ag)}\cdot (dg-ce) - \frac{\init_{<_2}(dg-ce)}{\mathrm{gcd}(dg,ag)}\cdot (ag-bf)  \\    
&=a(dg-ce) -d(ag-bf) \\
&= bdf-ace.
\end{align*} 

\noindent which is found by applying Buchberger's algorithm. This involves computing the $S$-polynomials between pairs of generators, finding the remainder after polynomial division by $\mathcal{G}$, and extending $\mathcal{G}$ by adjoining any non-zero remainders. This process is repeated until all remainders are $0$. For specifics about the algorithm, refer to \cite{CLO}. \hfill $\square$

\end{example}

Although Buchberger's algorithm provides a method for constructing a Gr\"obner basis for any given $I\subseteq R$ with respect to some given monomial order $<$, modern techniques have become more sophisticated and differ from the $S$-polynomial computation above \cite{BFS}.

A fundamental result in the theory of monomial ideals is Dickson's Lemma which states that every monomial ideal has a unique minimal monomial generating set. If some fixed monomial order of $R$ is also given, then every monomial ideal of $R$ will have a unique minimal \emph{ordered} generating set. 

\begin{lemma}\label{lem: unique monomial}
Given a monomial order $<$ on $R=\mathbb{K}[x_1,\ldots,x_n]$ and a monomial ideal $M\subseteq R$, there exists a unique minimal monomial ordered generating set of $M$.
\end{lemma}

We will be applying hash functions to sets of monomial generators, so having a unique way to write a given list of monomials is necessary. In defining the cryptosystem presented in the next section, we will also discuss sets of minimal monomial ideal generators which have bounded exponents, so we conclude this section with the following definition.

\begin{definition}\label{def: M_k}

We say that a set of monomials $\{\mathbf{x}^{\alpha_1},\ldots, \mathbf{x}^{\alpha_r}\}\subset R=\mathbb{K}[x_1,\ldots,x_n]$ is minimal if it is the unique minimal monomial generating set of the ideal $\langle \mathbf{x}^{\alpha_1},\ldots, \mathbf{x}^{\alpha_r}\rangle$. Then define $\mathcal{M}_k$ to be the set
\[\mathcal{M}_k := \big\{\{\mathbf{x}^{\alpha_1},\ldots, \mathbf{x}^{\alpha_r}\}\subset R\hspace{1mm}| \{\mathbf{x}^{\alpha_1},\ldots, \mathbf{x}^{\alpha_r}\}  \text{ is minimal}; \alpha_{i,j}< k, 1\leq i \leq r, 1\leq j\leq n\big\}.\]
\end{definition}

\subsection{Gr\"obner Fans} \label{sec: GrobFan}

Given an ideal $I\subset R=\mathbb{K}[x_1,\ldots, x_n]$, there is a formal combinatorial structure for enumerating all possible initial ideals of $I$. We first need to define what it means for two monomial orders to be equivalent, which is best done in the more general setting of weight orders. 

\begin{definition}
 Given a polynomial ring $R=\mathbb{K}[x_1, \ldots, x_n]$ and weight vector $w=(w_1, \ldots ,w_n) \in \mathbb{R}^n$, we can define a \emph{weight order} $<_w$ on $R$ by
 \begin{center}
$x_1^{a_1} \cdot \ldots \cdot x_n^{a_n} \leq_w x_1^{b_1} \cdot \ldots \cdot x_n^{b_n}$ \text{ if and only if } $a_1w_1+ \ldots + a_nw_n \leq_w b_1w_1+ \ldots + b_nw_n$.
 \end{center}
\end{definition}

\begin{remark}
For a fixed $I$ and monomial order $<$ on $R$, there exists a weight $w \in \mathbb{R}^n$ such that $\init_<(I)=\init_{<_w}(I)$. In fact, every monomial order is a weight order, but not every weight order is a monomial order. In particular, $\init_{<_w}(\cdot)$ does not necessarily yield a unique initial monomial term for every choice of $w$. As an example, let $f = x^2y+x^{10}-y^2 \in \mathbb{K}[x,y]$ and $w=(2,10)$. Then $\init_{<_w}(f) = x^{10}-y^2$.
\end{remark}

Given an ideal $I$, we can define an equivalence relation $\sim$ on $\mathbb{R}^n$ as $w_1\sim w_2$ if and only if $\init_{<_{w_1}}(I)=\init_{<_{w_2}}(I)$. If the initial ideal is a monomial ideal, then the collection of weight vectors in that equivalence class define a maximal cone in $\mathbb{R}^n$ and the union of these cones is called the \emph{Gr\"obner fan of I}, denoted by $GFan(I)$. We will refer to a maximal cone in $GFan(I)$ as a \emph{Gr\"obner region} (or \emph{Gr\"obner cone}), and denote it by  $GR_<(I)$. Here, $GR_<(I)$ has a representative weight order $<$ and corresponds to a distinct monomial initial ideal of $I$ together with a marked reduced Gr\"obner basis $\mathcal{G}_<$ for $I$. The union of these reduced Gr\"obner bases defines a universal Gr\"obner basis of $I$. See \cite{S} for more information about Gr\"obner fans, and for proofs of these facts. The Gr\"obner fan is the normal fan of a polytope \cite[Theorem 2.5]{S} called the \emph{state polytope of $I$}, and is denoted by $State(I)$. See \cite{CLS} for more information on correspondence between normal fans and convex polytopes.  

\begin{proposition}\cite[Corollary 1.3]{S}
Let $I\subset R=\mathbb{K}[x_1,\ldots,x_n]$, and suppose that  $\mathcal{G}_<\subset R$ such that $\init_<(\mathcal{G}_<)$ is the initial ideal of $I$ associated to the Gr\"obner region $GR_<(I)$. Then \[ \bigcup_{GR_<(I)\subseteq GFan(I)} \mathcal{G}_<\] defines a universal Gr\"obner basis of $I$.
\end{proposition}

With this structure, we now turn to the question of how to compute a Gr\"obner fan of an ideal $I$. Our protocol uses keys which are defined by initial ideals of $I$, so each Gr\"obner region of $GFan(I)$ defines a distinct key that can be used for encryption. Indeed, if an attacker with no knowledge about how $I$ or $\mathcal{U}_I$ were constructed wanted to compute $\mathcal{U}_I$ directly from $I$, they would need to first compute the Gr\"obner fan of $I$, and then take a union of all Gr\"obner basis representatives for the cones in the fan.

In \cite[Section 7]{S}, Sturmfels proposed an algorithm for computing $GFan(I)$, which was later implemented with the \texttt{Gfan} package \cite{J2} of \texttt{Macaulay2}. An analysis of the complexity of the algorithm was considered in \cite{FJT}. In this algorithm, there are three main steps that are iterated at each vertex of $State(I)$ (i.e. at each vertex of the graph of the polytope $State(I)$). We refer the reader to the aforementioned sources for the technical details of the algorithm, and only provide a very brief summary of the nature of each step below.

In the algorithm, we start with one known Gr\"obner region of the fan (e.g. by directly computing a Gr\"obner basis with respect to some order). We then travel from one maximal cone to a neighboring maximal cone while also computing a Gr\"obner basis representative for the new cone by a mutation of the Gr\"obner basis for the current one. 

\begin{enumerate}
   \item Each fan can be viewed as the normal fan to a dual object called a state polytope $State(I)$. The graph $(V,E)$ of the Gr\"obner fan is the $1$-skeleton of $State(I)$. Each vertex in the polytope corresponds to a cone in $GFan(I)$. Computing normal vectors to facets for this polytope requires time denoted by $T_{\mathrm{facets}}$.
   \item Starting at a vertex, we choose a sequence of admissible ``shoot'' edges in $State(I)$ until all vertices have been reached (using a reverse search algorithm). Selecting these admissible search edges requires time denoted by $T_{\mathrm{shoot}}$.
   \item With an admissible edge chosen, there exists an operation to change the marked Gr\"obner basis for one vertex to a marked Gr\"obner basis for another vertex. This is called a ``flip'' procedure and is the most time consuming step to complete. The time for this procedure is denoted by $T_{\mathrm{flip}}$.
\end{enumerate}

\begin{remark}\label{rmk: lattice}
We will later compare the complexity of computing $GFan(I)$ with other lattice-based primitives. Since the construction of $GFan(I)$ involves cones in $\mathbb{R}^n$, we can always choose a weight vector $w$ in each cone with only integer entries. This allows us to view $GFan(I)$ and $State(I)$ as lattice constructions by intersecting with $\mathbb{Z}^n$.
\end{remark}

\section{A quantum-resistant key establishment protocol} \label{sec: protocol}

In this section we will introduce a public encryption protocol that uses Gr\"obner bases for symmetric key establishment. The main idea is to leverage the difficulty in computing a universal Gr\"obner basis of a polynomial ideal compared to the computational time needed to compute a single Gr\"obner basis. We start by providing an example which illustrates the spirit of the algorithm, and in the subsequent section, we formally describe the protocol parameters and  implementation. Experts who wish to only consider the formal description of the protocol rather than a particular small-scale implementation may skip to Section \ref{sec: parameters}.

\subsection{An illustrative example} \label{sec: protocol example}

Let $A$ and $B$ be two parties who wish to securely communicate but have not previously communicated to share a common key. In this example, $B$ would like to send an encrypted message to $A$ using information publicly provided by $A$. Specific (simple) choices of protocol parameters have been chosen for this example. 

Suppose that $A$ has an ideal $I$ for which a universal Gr\"obner basis $\mathcal{U}_I$ is known. For this example, suppose that \[I=\langle ag-bf, ce-dg, ace-bdf \rangle \subset \mathbb{K}[a,\ldots,g]\]
\noindent is an ideal, where $\mathbb{K}$ is any field and $a,\ldots, g$ are indeterminates of the polynomial ring. This list of generators is kept private and provides Party $A$ with a fast way to compute initial ideals of $I$ (see Section \ref{sec: GrobFan} for a background on Gr\"obner fans). A second list of minimal generators $\mathcal{R}_{I}$ is also computed for $I$. For instance, $A$ may choose \[\mathcal{R}_{I} = \{ag-bf,ce-dg\}.\] It is worth noting that $ace-bdf=a(ce-dg)+d(ag-bf)$, so these generators still define the ideal $I$, but do not always form a Gr\"obner basis of it, depending on the monomial order. Party $A$ then makes $\mathcal{R}_{I}$, $\mathbb{K}$, and the number of variables public, together with information about the encryption scheme. 

Suppose that $B$ would like to send $A$ a secure message. $B$ proceeds to produce an encryption key using that public information that $A$ has provided. They do this by randomly selecting some monomial order $<_B$ of $\mathbb{K}[a,\ldots,g]$, and then proceeds to compute $\init_{<_B}\langle \mathcal{R}_I\rangle$. Both $<_B$ and the initial ideal are kept private. There is a unique minimal generating set for this monomial ideal by Lemma \ref{lem: unique monomial}, and based on a public hash function $\eta$ provided by $A$, converts this monomial ideal into an encryption key. For this example, $B$ chooses the lexicographic order $c>d>e>f>g>b>a$ so that \[\init_{<_B}\langle \mathcal{R}_I\rangle = \langle bf,ce\rangle.\] At this point, $B$ would apply $\eta$ to $\{bf,ce\}$. As an overly simplified method for generating an encryption key for this example, let us concatenate the exponent vectors of each minimal monomial generator (using a predefined order from the function provided by $A$). Then
\[bf = a^0b^1c^0d^0e^0f^1g^0  \longrightarrow 0100010\]
\[ce = a^0b^0c^1d^0e^1f^0g^0 \longrightarrow 0010100\]

\noindent Therefore, $B$'s private encryption key is $K_B=01000100010100$ which is used to encrypt a message using whatever encryption scheme $\mathcal{E}$ Party $A$ has made public. For instance, the encryption scheme may involve using $K_B$ to generate two large primes from a hash function attached to $\mathcal{E}$ and encrypt a message using the RSA encryption scheme (without making the product of the two primes public -- even to $A$). In the first version of the protocol, only this ciphertext is sent to $A$. 

With $\mathcal{U}_I$, $A$ will have the list of all possible initial ideals of $I$, and therefore will have a list of all possible encryption keys that $B$ could have generated. For this example, those keys are:
\begin{align*}
\langle ag, ce \rangle  & \longrightarrow 10000010010100\\  
\langle ag, bdf, dg \rangle &\longrightarrow 100000101010100001001\\
\langle ace, ag, dg \rangle  &\longrightarrow 101010010000010001001\\
\langle bf, ce \rangle  &\longrightarrow 01000100010100\\
\langle bf, dg \rangle  &\longrightarrow 01000100001001\\
\end{align*}

\noindent Party $A$ then tries each key (by a brute-force search) until the message is decrypted. Using some symmetric encryption algorithm $\mathcal{E}$ instead of an asymmetric one would reduce the computational time for Party $A$. 

In an alternate version of the protocol, a hash function $\tau$ is also made public that is used to map $K_B$ to some binary sequence. Then, together with the ciphertext, Party $B$ also sends $\tau(K_B)$. Party $A$, having all possible values of $\tau$ applied to each key, would then know which $K_B$ was chosen, and could decrypt the message immediately. This comes at the cost of reduced security for $\mathcal{P}$.

We will now summarize each step of the protocol and briefly highlight security issues which are elaborated on in Section \ref{sec: security}:

\begin{enumerate}
    \item Party $A$ is assumed to have a universal Gr\"obner basis $\mathcal{U}_I$ for some ideal $I\subset \mathbb{K}[x_1\ldots, x_n]$ and some trimmed generating set $\mathcal{R}_I$. This trimmed list is made public, together with a list of conventions needed for $B$ to establish an initialization vector. A method for $A$ to quickly find examples of $I$ and $\mathcal{U}_I$ can be found in Section \ref{sec: initial}. The complexity of computing $\mathcal{R}_{I}$ can be found in Theorem \ref{thm: trim ideal}.
    \item If $B$ wants to send a message to $A$, then a monomial order $<_B$ needs to be chosen, and the initial ideal $\init_{<_B}\langle\mathcal{R}_I\rangle$ needs to be computed. This can be done relatively quickly, as described in Theorem \ref{thm: GB complexity}. By Lemma \ref{lem: unique monomial}, there exists a unique minimal generating set of $\init_{<_B}\langle\mathcal{R}_I\rangle$. $B$ uses this unique generating set to produce an encryption key using a function $\eta$  provided by $A$. The  message is then encrypted using the encryption scheme $\mathcal{E}$ (also publicly provided by $A$). 
    \item If $\tau$ is not provided, then Party $B$ only sends the ciphertext to $A$, keeping all other parameters private. If an attacker intercepts the message, they would not know which encryption key was used to create the message, and reducing the list of possible keys to a manageable list would require knowledge of $\mathcal{U}_I$. An analysis of the complexity of doing this and the general security of the system can be found in Section \ref{sec: quantum}. If $\tau$ is provided, then Party $B$ would also send $\tau(K_B)$ together with the ciphertext.
    \item Since $A$ has a list of all possible initial ideals of $I$ (which defines a Gr\"obner fan -- see Section \ref{sec: GrobFan}), they also have a list of all possible encryption keys. The message can then be decrypted by trying each one (in the case that $\tau$ is not provided). The time complexity for this is described in Section \ref{sec: complexity}. Variations that help cut down $A$'s decryption time are also discussed there. In the case when $\tau$ is provided, Party $A$ could decrypt the ciphertext without iterating over all keys.
\end{enumerate}

This protocol relies on the fact that it is both time consuming to enumerate all potential encryption keys and to directly try and generate a universal Gr\"obner basis for $I$. Therefore, an attacker would need to spend an unreasonable amount of time and resources to decrypt the message, unless they already had a universal Gr\"obner basis for $I$. If this is the case, then how is it possible for $A$ to even initialize the system? Using a combinatorial approach, we will show how to efficiently generate examples with symmetries that are kept private but for which $\mathcal{U}_I$ is easy to compute.

\subsection{Protocol Parameters} \label{sec: parameters}

As seen in the previous section, there are a number of parameters which can be chosen when setting up this cryptosystem. We provide the theoretical framework for this protocol and leave choices of specific parameters unspecified, recognizing that there may be practical implementation concerns that may require flexibility. 

\begin{definition}\label{def: par}
Let $\mathcal{P}= (n,\mathbb{K}, \mathcal{U}_I, \mathcal{R}_I, <_B, \eta, \mathcal{E},\tau)$ be a tuple of parameters defined by:

\begin{itemize}
    \item $n \in\mathbb{N}$ and $\mathbb{K}$ is a field.
    \item $\mathcal{U}_I$ is a universal Gr\"obner bases for an ideal $I\subset \mathbb{K}[x_1,\ldots,x_n]$. 
    \item $\mathcal{R}_I$ is some (usually minimal) generating set of $I$. 
    \item $<_B$ is a monomial order on $\mathbb{K}[x_1,\ldots,x_n]$.
    \item $\eta: \mathcal{M}\rightarrow \{0,1\}^s$ is a hash function on sets of monomials whose domain contains $\mathcal{M}_k$ (see Definition \ref{def: M_k}) to binary strings of fixed length $s$.
    \item $\mathcal{E}$ is an encryption algorithm with an encryption key $K$ and plaintext $X$ as input, and whose output is a ciphertext $\mathcal{E}(K,X)$.
    \item $\tau$ is either a hash function on binary sequences $\{0,1\}^s$, or is set to $\emptyset$ if not being utilized.
\end{itemize}
\end{definition}

Let $A$ and $B$ be two parties who wish to securely communicate but have not previously communicated to share a common key. 

\begin{enumerate}[(I)]
    \item \textbf{Initialization:} Party $A$ generates an ideal $I\subset \mathbb{K}[x_1,\ldots, x_n]$ and a list of polynomials $\mathcal{U}_I\subset \mathbb{K}[x_1,\ldots, x_n]$ which define a universal Gr\"obner basis of $I$. A reduced generating set $\mathcal{R}_I$ for $I$ is also computed. The set  $\mathcal{U}_I$ is kept private. The information $(n,\mathbb{K}, \mathcal{R}_I, \eta, \mathcal{E},\tau)$ is made public, and is sent to $B$: \[ A\longrightarrow B: \mathcal{P}_{\mathrm{public}}=(n,\mathbb{K}, \mathcal{R}_I, \eta, \mathcal{E},\tau).\]
    \item \textbf{Key Generation:} $B$ uses  $\mathcal{P}_{\mathrm{public}}$ to generate a private key $K_B$ by first selecting a monomial order $<_B$ on $\mathbb{K}[x_1,\ldots,x_n]$. $B$ then computes $\init_{<_B}\langle \mathcal{R}_I\rangle$ which has a unique minimal monomial generating set $\mathcal{G}$. Then $K_B = \eta(\mathcal{G})$, which is kept private. 
    \item \textbf{Encryption:} $B$ can now encrypt any plaintext $X_B$ using $\mathcal{E}$ and $K_B$. If $\tau=\emptyset$, only this ciphertext is sent to $A$: \[B\longrightarrow A: \mathcal{E}(K_B,X_B).\] For a protocol defined with $\tau\neq\emptyset$,  the value $\tau(K_B)$ is also sent:  \[B\longrightarrow A: (\mathcal{E}(K_B,X_B),\tau(K_B)).\]
    \item \textbf{Decryption:} Since $A$ has a universal Gr\"obner basis $\mathcal{U}_I$, they possess all possible initial ideals $\init_<(I)$, and hence possess all possible encryption keys $\mathcal{K_I}$. If $\tau=\emptyset$, then for each $\kappa\in\mathcal{K}_I$, $A$ can compute $\mathcal{E}^{-1}(\kappa, \mathcal{E}(K_B,X_B))$ until the message is decrypted by brute force (choosing $\mathcal{E}$ to be a symmetric encryption algorithm will reduce this computational time). 
    For a protocol defined with $\tau\neq\emptyset$, Party $A$ also has the list $\tau(\mathcal{K}_I)$, and hence knows which $K_B$ was used to produce $\tau(K_B)$, and can therefore immediately compute $\mathcal{E}^{-1}(K_B, \mathcal{E}(K_B,X_B))$.
\end{enumerate}

\begin{remark} When $\tau=\emptyset$ and Party $A$ needs to iterate over all encryption keys, there is a question of how they will know when the ciphertext has been correctly decrypted. Besides being a readable message, one method is to require the message to contain some embedded marker  defined from $\mathcal{E}$.
\end{remark}

With this language, there are infinitely many cryptosystems that can be initialized, depending on the parameters of $\mathcal{P}$. We will discuss the security of this system and how to compute each component in Section \ref{sec: security}. For now, we show that even if a breach were to occur where $\mathcal{U}_I$ were made public, the cryptosystem still provides a reasonable amount of security.

\begin{proposition}\label{prop: data breach}
Suppose that $\mathcal{P}= (n,\mathbb{K}, \mathcal{U}_I, \mathcal{R}_I, <_B, \eta, \mathcal{E},\tau)$ is a set of protocol parameters as in Definition \ref{def: par} such that $\mathcal{U}_I=\mathcal{R}_{I}$. Denote $r=|\mathcal{U}_I|$ and let $m$ be the maximum number of monomial terms of a polynomial in $\mathcal{U}_I$. Then \[\#\text{ of possible encryption keys of $\mathcal{P}$ } \leq m^r.\]

\noindent In particular, the number of operations needed to conduct a brute-force attack of a system where $\mathcal{U}_I$ is public is $O(m^r)$.
\end{proposition}
\begin{proof}
Any potential initial ideal of $I$ can be read off from $\mathcal{U}_I$ by picking one term from each $f\in \mathcal{U}_I$. There are at most $m$ such terms for each $f$, leaving at most $m^r$ many possibilities. 

The remaining operations of converting each list of monomials to an encryption key via $\eta$, and decrypting the message using $\mathcal{E}^{-1}$ is some constant multiple of $m^r$ (for bounded key and message sizes).
\end{proof}

\begin{remark}
Note that each potential list of monomials from the proof of Proposition \ref{prop: data breach} may not produce a valid initial ideal since there may not exist a monomial order of $\mathbb{K}[x_1,\ldots, x_n]$ which yields that particular list of monomials (for example, there is no monomial order which picks $x_1x_2$ as an initial term of $x_1^2+x_1x_2+x_2^2$). We also saw this in Section \ref{sec: protocol example} where there were a priori $2^3$ possible lists of monomials, but only 5 defined actual keys.
\end{remark}

In another extreme, suppose that $\mathcal{U}_I$ remains private and an attacker wanted to generate all potential encryption keys for a system $\mathcal{P}$ without computing a universal Gr\"obner basis for $I$. They would need to generate all encryption keys using $\eta$ associated to $\mathcal{M}_k$. With the assumption that the allowed exponent vectors have entries strictly smaller than $k$, the number of possible encryption keys is a doubly exponential in $n$. Even in the square-free case where $k=2$, a brute-force attack would be unwieldy.

\begin{proposition}\label{prop: max keys}
Let $\mathcal{P}= (n,\mathbb{K}, \mathcal{U}_I, \mathcal{R}_I, <_B, \eta, \mathcal{E},\tau)$ be as in Definition \ref{def: par} such that $\eta$ is a map whose domain contains $\mathcal{M}_k$. Then \[\#\text{ of possible encryption keys of $\mathcal{P}$ } \leq 2^{k^n}.\]

\noindent In particular, the number of operations needed to conduct a brute-force attack of a system $\mathcal{P}$ is $O(2^{k^n})$.
\end{proposition}
\begin{proof}
There are exactly $k^n$ monomials in $\mathbb{K}[x_1, \ldots, x_n]$ whose exponent vector entries are bounded above by $k$. A minimal generating set of a monomial ideal is a subset of this list of monomials, so the number of possible subsets is $2^{k^n}$. Different choices of a subset may result in the same monomial ideal, but each subset results in one case that needs to be checked. The complexity bound follows similarly to the proof of Proposition \ref{prop: data breach}.
\end{proof}

\begin{remark}
An attacker could also try to compute all possible initial ideals of $I$ by  enumerating all monomial orders on $\mathbb{K}[x_1,\ldots, x_n]$ and computing a Gr\"obner basis directly for each one. Using only the lexicographic orders shows that this approach has at least $n!$ cases to check (there are also other monomial orders which are not equivalent to a lexicographic order). This, together with the amount of time needed to perform one Gr\"obner basis computation would make this approach untenable. 
\end{remark}

\section{Security Analysis} \label{sec: security}

The security of the protocol in the last section depends on the fact that it is generally difficult to compute the universal Gr\"obner basis for an ideal. Much work has been done to describe procedures for this computation (for example, see \cite{FJT,S}), and there exist programs to compute such a basis. However, even for small ideals, the computational complexity can be a barrier to finding an explicit list. For example, let $\mathrm{Det}_{t,m,n}$ denote the ideal in the polynomial ring in $mn$ variables generated by the $t \times t$ minors of an $m \times n$ generic matrix. For $\mathrm{Det}_{3,4,4}$ alone, there are over 160,000 possible initial ideals \cite{FJT}, so a brute-force universal Gr\"obner basis construction would require over 160,000 individual (albeit simplified) Gr\"obner basis computations. 

This section is dedicated to showing that an attacker with no prior knowledge of how $\mathcal{U}_I$ was constructed would require an unreasonable amount of time to compute $\mathcal{U}_I$. In Section \ref{sec: initial}, we contrastingly show that Party $A$ can construct $\mathcal{U}_I$ with relative ease if $I$ is chosen with some symmetries. In this section, we will also show that Party $B$ does not require much time to encrypt a message, and that the required time scales well with any additional complexity that $A$ adds to the system $\mathcal{P}$.

\subsection{Quantum Resistance} \label{sec: quantum}

In Section \ref{sec: GrobFan}, we outlined the steps needed to compute $GFan(I)$, the Gr\"obner fan of $I$, and by extension the state polytope of $I$, $State(I)$. Each maximal cone in $GFan(I)$ (or dually, any vertex of $State(I)$), defines one possible initial ideal of $I$, and hence defines a possible encryption key that Party $B$ could have used. Based on how Party $A$ decrypts messages using $\mathcal{P}$ when $\tau=\emptyset$, a rational attacker without any prior knowledge of the symmetries of $I$ would need to compute $GFan(I)$ to guarantee the decryption of the message (see Propositions \ref{prop: data breach} and \ref{prop: max keys} for other less effective attacks).

The algorithm proposed in \cite{FJT} is a practical implementation of the theoretical algorithm proven in \cite[Section 7]{S}. It involves three main steps that are iterated at each vertex of $State(I)$ (i.e. on the graph of the polytope $State(I)$). These steps are technical, and are only briefly outlined in Section \ref{sec: GrobFan}. The complexity of performing these steps was analyzed in \cite{FJT} and is highlighted in the next theorem. It ultimately shows that an attacker would require an inordinate amount of resources to compute the Gr\"obner fan directly, assuming $I$ is sufficiently complex.

\begin{theorem}\label{thm: Gfan} Let $(V,E)$ be the graph of the Gr\"obner fan of $I$. The time complexity for computing this graph given a marked reduced Gr\"obner basis is the class of functions

\[
O\biggl(\hspace{1mm}\sum_{\mathcal{G} \in V}T_{\mathrm{facets}}(\mathcal{G}) + \sum_{(\mathcal{G}_1,\mathcal{G}_2) \in E}T_{\mathrm{shoot}}(\mathcal{G}_1)+ \sum_{(\mathcal{G}_2,\mathcal{G}_1)\in E}T_{\mathrm{flip}}(\mathcal{G}_1,\mathcal{G}_2)\biggr),
\]
which can be written as
\[
O\biggl(\hspace{1mm}\sum_{\mathcal{G} \in V}T_{\mathrm{lp}}(n,r(\mathcal{G}))r(\mathcal{G}) + \sum_{(\mathcal{G}_1,\mathcal{G}_2) \in E}r(\mathcal{G}_1)n^2+ T_{\mathrm{lp}}(n,\mathcal{G}_1)+ \sum_{(\mathcal{G}_2,\mathcal{G}_1)\in E}T_{\mathrm{flip}}(\mathcal{G}_1,\mathcal{G}_2)\biggr),
\]
where $r(G)$ is the number of non-leading terms in the marked reduced Gr\"obner basis $G$, and $T_{\mathrm{lp}}$ is the time complexity related to finding an interior point to a cone. The first two terms are bounded by a polynomial in the size of the output. The computational complexity of the third term is NP-hard.
\end{theorem}
\begin{proof}
The complexity description and bound for the first two terms can be found in \cite[Theorem 5.1]{FJT}. The statement about the third term can be found in \cite[Section 5]{J}, using a result from tropical geometry. More specifically, a tropical variety is a union of Gr\"obner regions, making it a subfan of a Gr\"obner fan. In \cite[Section 3]{T}, several decision problems related to tropical varieties (which could be solved given a Gr\"obner fan \cite[Remark 3.4]{T}) are shown to be NP-hard.
\end{proof}

\begin{example}\label{ex: Gfan}
Let $\mathrm{Det}_{t,m,n}$ denote the ideal in $\mathbb{K}[x_{11},\ldots, x_{mn}]$ generated by the $t \times t$ minors of the matrix:
\begin{align*}
\begin{pmatrix}
x_{11} & x_{12} & \ldots & x_{1n}\\
x_{21} & x_{22} & \ldots & x_{2n}\\
\vdots & \vdots & \ddots & \vdots \\
x_{m1} & x_{m2} & \ldots & x_{mn}
\end{pmatrix} .
\end{align*}     

Using the \texttt{Gfan} software package \cite{J2} of \texttt{Macaulay2}, the authors in \cite{FJT} were able to show that on a standard 2.4 GHz Pentium processor at the time (i.e. 2005), $\mathrm{Det}_{3,4,4}$ has $163,032$ many Gr\"obner regions in its Gr\"obner fan, and hence $163,032$ possible initial ideals. Note that $\mathrm{Det}_{3,4,4}$ is an ideal that initially is only generated by 16 polynomials in 16 variables. Without the use of the symmetries of $\mathrm{Det}_{3,4,4}$, the full computation took approximately 14 hours. Using the symmetries of $\mathrm{Det}_{3,4,4}$, the computation time for the full-dimensional cones took only 7 minutes. While a modern computer could do this considerably faster, the number of computations remains unchanged, so scaling $m$ and $n$ would produce similar results. \hfill $\square$

\end{example}

In light of the recent threats that quantum computers pose to traditional public-key cryptosystems, it becomes increasingly important to develop new techniques to secure data that can resist attacks from a quantum computer. Post-quantum cryptography is a relatively new field, but there are six main classes of algorithms which are considered to be resistant to quantum attacks \cite{B}. Therefore, to show that our protocol $\mathcal{P}$ is quantum-resistant, it suffices to demonstrate that it belongs to one of these categories. 

Since $State(I)$ has vertices in $\mathbb{Z}^n$ \cite[Chapter 2]{S}, it can be viewed as a convex polytope in the standard integer lattice $\mathbb{Z}^n$. Computing $GFan(I)$ directly is NP-hard by Theorem \ref{thm: Gfan}, and would generally require millions of simplified Gr\"obner basis calculations, in addition to polytope computations (like finding normal vectors to the facets, a Gr\"obner walk through the graph of $State(I)$, etc.). Therefore, these structures would be considered lattice-based primitives. Lattice-based cryptography is one of the previously mentioned approaches considered resistant against quantum attacks \cite{B}. To prove the security of $\mathcal{P}$, we will need to show that computing $GFan(I)$ is at least as difficult as a lattice-based problem which is considered quantum-resistant. Additionally, if $\tau$ is used, then it needs to be chosen to be quantum-resistant too.

\begin{theorem}\label{thm: secure}
Assume that lattice-based primitives secured by the Shortest Vector Problem (SVP) are quantum-resistant. If $\tau\neq \emptyset$, then also assume that $\tau$ is quantum-resistant. Then a protocol $\mathcal{P}$ for which $|\mathcal{U}_I|$ is sufficiently large is quantum-resistant. 
\end{theorem}
\begin{proof}
We need to demonstrate that computing a Gr\"obner fan is as difficult as solving a certain SVP. To do this, we will show that when properly rephrased, knowing  $\mathcal{U}_I$ will provide a solution to a certain SVP. 

Recall that $GFan(I)$ is a union of cones, called Gr\"obner regions, and in the interior of each cone $GR_<(I)$, there is at least one integer lattice point $x \in GR_<(I)\cap\mathbb{Z}^n$ such that $|x|\leq |y|$ for all $y\in \mathrm{int}(GR_<(I)\cap\mathbb{Z}^n)$ where $|\cdot|$ is the usual Euclidean norm on $\mathbb{R}^n$. Consider the following question: Among all Gr\"obner regions $GR_<(I)\subset GFan(I)$, which one has an interior integer lattice point $x\in GR_<(I)\cap\mathbb{Z}^n$ which is closest to the origin?

Suppose that $GFan(I)$ is known. Then the normal vectors to facets of $State(I)$ have been computed, and therefore lattice generators for each rational polyhedral cone $GR_<(I)\cap\mathbb{Z}^n$ are known. For simplicial cones, given integer vectors $v_1,\ldots, v_n$ that generate the cone $GR_<(I)\cap\mathbb{Z}^n$, the interior points have the form $c_1v_1+\ldots+c_nv_n$ with $c_i>0$ and $c_i\in \mathbb{Z}$. The interior lattice point with the closest distance is precisely $v_1+\ldots+v_n$ (recall that the cones computed in the algorithm are contained in the positive orthant \cite[Definition 2.8]{FJT}). The non-simplicial case can be computed using integer linear programming techniques for rational convex polytopes, which runs in polynomial time for a fixed dimension \cite{DLHTY}. In particular, $GFan(I)$ has more information than what is needed to answer the question posed in the last paragraph, and the question can be answered in polynomial time.

On the other hand, given $GFan(I)$, we could take the collection of all appropriately scaled lattice vectors generating all Gr\"obner regions, and select some minimal generating set $\mathcal{C}$ from this collection. Note that the shortest vector that answers the above question is a rational combination of some of these generators. By scaling, we can assume that the rational combination yields an integer combination. The (now scaled) vectors in $\mathcal{C}$ generate a lattice, and we can ask what the shortest vector is in that lattice, which is a specific independent SVP. The solution to this SVP is exactly equal to the vector computed in the previous paragraph, showing that computing $GFan(I)$ is at least as difficult as an SVP. 

Finally, when $\tau\neq\emptyset$, an attacker could try to compute $\tau^{-1}(\tau(K_B))$ directly, circumventing the security that $GFan(I)$ provides, so $\tau$ needs to be a hash function which is resistant to quantum attacks.
\end{proof}

\begin{remark}
There are additional ways in which computing $GFan(I)$ emulates the spirit of other lattice-based problems used in post-quantum security. For example, given a vertex $\mathbf{v}\in\mathbb{Z}^n$ of $State(I)$, as a face of the polytope, is characterized by the inequality $w\cdot\mathbf{v} > w\cdot \mathbf{u}$ for all other $\mathbf{u}\in State(I)$, where $w$ is the weight order associated to the vertex $v$. There is another vertex $v'$ of $State(I)$ which has the furthest distance from $v$, and is precisely the point of $State(I)$ which minimizes the functional $f(P)=w\cdot P$ over $State(I)$. 

Many algebraic problems haven't been studied in the context of post-quantum cryptography, so many of the established lattice-based problems that have been formulated do not immediately translate to algebraic settings. Further research is needed to establish independent algebraic or combinatorial problems which are considered quantum resistant.
\end{remark}

\subsection{Other complexities associated with $\mathcal{P}$} \label{sec: complexity}

We will start with the complexity of the one-time computation of $\mathcal{R}_{I}$ by Party $A$. There are numerous ways to trim the ideal $I$, given that $\mathcal{U}_I$ is already known. One possibility is to choose some monomial order $<$ of $\mathbb{K}[x_1,\ldots,x_n]$ and compute a Gr\"obner basis for $I$ using $\mathcal{U}_I$. This also generates the ideal $I$, and generally involves far fewer than $|\mathcal{U}_I|$ many generators.

\begin{theorem}\label{thm: trim ideal}%% That is, it is easy for Alice to trim the ideal
Let $\mathcal{U}_I$ be a universal Gr\"obner basis for an ideal $I\subset R=\mathbb{K}[x_1,\ldots,x_n]$ and $<$ some monomial order on $R$. Then a Gr\"obner basis for $I$ with respect to $<$ is one possible choice of $\mathcal{R}_I$. Furthermore, if $|\mathcal{U}_I|=N$, then $\mathcal{R}_I$ can be computed with $O(N)$ many operations. 
\end{theorem}
\begin{proof}
Since $\mathcal{U}_I$ is a universal Gr\"obner basis, selecting Gr\"obner generators of $I$ for a fixed $<$ is a simple procedure which involves marking the initial terms of each of the $N$ many elements of $\mathcal{U}_I$, and then eliminating redundancies and any elements whose initial terms are not needed to generate $\init_<(I)$.  
\end{proof}

Next we shift to the complexity of Party $B$ finding a Gr\"obner basis of $I$. In not having access to $\mathcal{U}_I$, $B$ will need to compute a Gr\"obner basis directly. The next theorem provides the complexity of this computation.

\begin{theorem}\label{thm: GB complexity}
\cite[Proposition 1]{BFS}    
Let $(f_1, \ldots, f_m)$ be  a system of homogeneous polynomials in $\mathbb{K}[x_1, \ldots, x_n]$ with $\mathbb{K}$ an arbitrary field. The number of operations in $\mathbb{K}$ required to compute a Gr\"obner basis of the ideal $I$ generated by $(f_1, \ldots, f_m)$ for a graded monomial ordering up to degree $D$ is bounded by 
\begin{align*}
O \biggl(mD \binom{n+D-1}{D}^\omega \biggr), \text{ as } D \rightarrow \infty
\end{align*}
where $\omega$ is the exponent of matrix multiplication over $\mathbb{K}$.
 
\end{theorem}

Recall that the exponent of matrix multiplication is a constant $\omega$, depending on $n$, which is used to bound the complexity of matrix multiplication of $n\times n$ matrices. Suppose that Party $A$ has already chosen a preset number of variables, so that $n$ is fixed, as well as the number of elements in a generating set (which is $m$). If we allow the degrees of the $f_i$'s to vary, then we are allowing the degree of the polynomials used in the Gr\"obner basis computation to vary, and thus the maximum degree of the Gr\"obner basis elements (i.e. the variable $D$) may also vary. Then, as a function of $D$,

\begin{align*}
\binom{n+D-1}{D} &= \frac{(n+D-1)(n+D-2)\cdots (D+1)D!}{D!(n-1)!}\\
&= \frac{(n+D-1)(n+D-2) \cdots (D+1)}{(n-1)!}\\
\end{align*}
which is a polynomial of degree $n-1$ in $D$. Therefore, 

\begin{align*}
mD \binom{n+D-1}{D}^\omega &= mD\left(\frac{(n+D-1)(n+D-2)\cdots (D+1)D!}{D!(n-1)!}\right)^\omega\\
&= O(D^{\omega(n-1)+1})
\end{align*}

Even if $A$ changes the complexity of a protocol $\mathcal{P}$ based on the degree of the generators considered (e.g. changes $D$ to $D+1$), it only affects $B$'s Gr\"obner basis computation polynomially.

\begin{corollary}\label{cor: GB}
With the same notation as Theorem \ref{thm: GB complexity}, let $n$ and $m$ be fixed positive integers. Then the number of operations in $\mathbb{K}$ required to compute a Gr\"obner basis of the ideal $I$ is polynomial in $D$. More precisely, it is bounded by \[ O(D^{\omega(n-1)+1}) , \text{ as } D \rightarrow \infty\]
\end{corollary}

\begin{remark}
By a theorem of Dub\'e, the maximum degree $D$ of polynomials appearing in a Gr\"obner basis is bounded by $(d^2+2d)^{2^{n-1}}$ where $d$ is the maximum degree of the polynomials in $\{f_1,\ldots,f_m\}$ (see \cite{D}). Note that this bound grows very large even for small $n$ and $d$, but in our case, Party $B$ is not computing a random Gr\"obner basis for a random ideal. On the contrary, the maximum degree of an element in $\mathcal{U}_I$ is an upper bound for the $D$ that Party $B$ would encounter. 
\end{remark}

On the other hand, since $\omega$ is dependent on $n$, the complexity of computing a Gr\"obner basis is doubly exponential in $n$. This complexity assumes the worst case however. It is expected that only a mild increase in the complexity will result from adding new generators with a similar structure to the current $f_i$ (for instance, adding binomial generators). In summary, Theorem \ref{thm: GB complexity} tells us that:

\begin{itemize}
\item An increase in the maximum degree of the generators $f_1,\ldots, f_m$ results in a polynomial increase in time for $B$. 
\item An increase in the number of generators $m$ for fixed $D$ and $n$ also results in a polynomial increase in time for $B$ (with complexity $O(m)$). 
\item An increase in $n$ could result in a doubly exponential jump in the worst case for Party $B$'s computation time \cite[Section 5]{J}, so care in selecting $I$ should be taken. 
\end{itemize}

We conclude this section with a brief statement about the maximum time needed for Party $A$ to decrypt a message when $\tau=\emptyset$.

\begin{proposition} \label{prop: decrypt}
Let $\mathcal{P}= (n,\mathbb{K}, \mathcal{U}_I, \mathcal{R}_I, <_B, \eta, \mathcal{E},\tau)$. Suppose that $N$ is the number of Gr\"obner regions in the Gr\"obner fan $GFan(I)$, and let $C(\mathcal{E})$ be the maximum amount of time needed to decrypt a ciphertext (of bounded length) using $\mathcal{E}$. If $\tau=\emptyset$, then the maximum amount of time needed to decrypt a message (of bounded length) sent using $\mathcal{P}$ is $N\cdot C(\mathcal{E})$. 
\end{proposition}

\begin{example}
Let us continue with Example \ref{ex: Gfan} when $\tau=\emptyset$. We will let $\mathcal{E}$ be the RSA encryption/decryption scheme. Suppose that the decryption time using $\mathcal{E}$ is approximately  $C(\mathcal{E}) \sim 0.00055$ seconds on a standard laptop. For this example, Party $A$ would have to check at most $N=163,032$ keys, taking at most 90 seconds to decrypt the message. Having partial information from Party $B$ in the message about the key $K_B$ would reduce this time significantly and may be needed to bring down the decryption time to a more reasonable number, for practical purposes.

It is worth noting that in this example, Party $B$ would only take about 0.3 seconds to compute a single Gr\"obner basis for a fixed order $<_B$ (found by dividing the 14 hours needed to compute the Gr\"obner fan by the number of Gr\"obner regions). \hfill $\square$
\end{example}

\begin{remark}
When $\tau=\emptyset$, the protocol $\mathcal{P}$ offers greater security  since no information about $K_B$ is sent publicly, but this comes at a cost to Party $A$ who now needs to iterate through all possible encryption keys. Choosing a symmetric $\mathcal{E}$ can help reduce this time. On the other hand, if $\tau\neq \emptyset$, then Party $A$ can decrypt a ciphertext very quickly, at the cost of information about $K_B$ being public using $\tau$. In this case, the security of the system is also dependent on the security of $\tau$. 
\end{remark}
\section{Effective Initialization of $\mathcal{U}_I$} \label{sec: initial}

In the last section, we saw that the problem of an attacker trying to compute $\mathcal{U}_I$ directly is intractable if the set is sufficiently large, leading one to wonder how it is possible for Party $A$ to even initialize the system $\mathcal{P}$. Here we will introduce a way to easily compute examples of $\mathcal{U}_I$ using toric ideals of graphs. The knowledge of which graph was used in the construction would provide a trapdoor to an attacker computing $\mathcal{U}_I$, so it is imperative that this information be kept private. The content of this section is also of independent interest to combinatorial algebraists, especially those working with toric ideals of graphs and geometric vertex decomposition (see \cite{DLST} for example).

\subsection{The toric ideal of a graph}

Let $G=(V(G), E(G))$ be a finite simple graph where $V(G) = \{v_1, \ldots, v_n\}$ is the set of vertices of $G$, and $E(G) = \{e_1, \ldots, e_r\}$ is the set of edges of $G$ with $e_i = \{v_{i_j},v_{i_k}\}$ an unordered pair of vertices which we call the endpoints of $e_i$. Given $G$, we can associate an ideal $I_G$ to it. Let $\mathbb{K}[E(G)] = \mathbb{K}[e_1, \ldots, e_r]$ and $\mathbb{K}[V(G)] = \mathbb{K}[v_1, \ldots, v_n]$ be two polynomial rings over $\mathbb{K}$ with the edges and vertices viewed as indeterminates, respectively. Then consider the $\mathbb{K}$-algebra homomorphism,
\begin{align*}
\varphi_G : \mathbb{K}[E(G)] \rightarrow \mathbb{K}[V(G)]
\end{align*}
defined on the indeterminates $e_i$ by $\varphi_G(e_i)=v_{i_j}v_{i_k}$ where $e_i = \{v_{i_j},v_{i_k}\}$ for all $i \in \{1, \ldots, r\}$. The kernel of the map $\varphi_G$ will be denoted by $I_G$ and is called the \emph{toric ideal of the graph $G$}.

There is a convenient graph-theoretic description of the elements of $I_G$. First, recall that a \emph{walk} of length $k$ in a graph $G$ is an alternating sequence of vertices and edges \[W = (v_{i_0},e_{i_1}, v_{i_1}, e_{i_2}, v_{i_2}, \ldots, v_{i_{k-1}}, e_{i_k}, v_{i_k}) \] where $e_{i_j} = \{v_{i_{j-1}},v_{i_{j}}\}$ for $j=1, \ldots, k$. We say that the walk is even if $k$ is even, and closed if $v_{i_k}=v_{i_0}$. We can associate a binomial in $\mathbb{K}[E(G)]$ to $W$ by $e_{i_1}e_{i_3} \cdots e_{i_{k-1}}-e_{i_2}e_{i_4} \cdots e_{i_{k}}$. In general, all binomials associated to closed even walks of $G$ are in $I_G$. It turns out that these binomials generate $I_G$. 

\begin{theorem} \cite[Proposition 10.1.5]{V}
Let $G$ be a finite simple graph. Then the toric ideal $I_G$ of $G$ is generated by the set of binomials 
\begin{align*}
\{e_{i_1}e_{i_3} \cdots e_{i_{k-1}}-e_{i_2}e_{i_4} \cdots e_{i_{k}} | (e_{i_1}, \ldots, e_{i_{k}}) \text{ is a closed even walk of } G \}.
\end{align*}
\end{theorem}

There are generally infinitely many closed even walks of a graph $G$. To achieve a finite generating set, we consider only \emph{primitive closed even walks}.

\begin{definition}
Let $e^A := e_1^{A_1}\cdots e_r^{A_r}$. A binomial $e^\alpha - e^\beta \in I_G$ is called \emph{primitive} if there is no other binomial $e^\gamma - e^\delta \in I_G$ such that $e^\gamma | e^\alpha$ and $e^\delta | e^\beta$. 
\end{definition}

Not only do the set of primitive closed even walks generate the ideal $I_G$, they are also a universal Gr\"obner basis of $I_G$.

\begin{theorem} \cite[Proposition 10.1.9]{V}
Let $G$ be a finite simple graph. Then the set of all primitive binomials of $I_G$ define a universal Gr\"obner basis of $I_G$, denoted by $\mathcal{U}(I_G)$.  
\end{theorem}

By taking $I=I_G$ for some $G$, we would automatically have a convenient description for $\mathcal{U}_I=\mathcal{U}(I_G)$. Even with this description, computing $\mathcal{U}(I_G)$ directly can be computationally difficult. In fact, the number of elements of $\mathcal{U}(I_G)$ can grow very quickly. For instance, $\mathcal{U}(I_{K_8})$ has over 40,000 elements \cite{DLST}. A description of how to compute this set can be found in \cite[Section 7]{S}. A graph-theoretic characterization of primitive closed even walks of a graph can be found in \cite{TT}.

\subsection{Generating large graphs}

In this section, we will show that it is possible to generate graphs for which $\mathcal{U}(I_G)$ is recursively computable, and for which $|\mathcal{U}(I_G)|$ is sufficiently large to ensure the security of the protocol $\mathcal{P}$. Furthermore, this can be done in polynomial time, depending on the number of constructive steps detailed below. We are going to introduce three operations for this purpose. 

\subsubsection{Gluing along a vertex} \label{subsubsec: vertex gluing}

Given a graph $G$, we can glue a disjoint graph $H$ to $G$ along a vertex by selecting some $v_G\in V(G)$ and $v_H\in V(H)$ and identifying the two vertices. More specifically, we define a new graph, denoted $G\star_{v_G,v_H} H$ (or simply $G\star H$ when $v_G$ and $v_H$ are understood), constructed as a disjoint union of the two graphs modulo the relation where $v_G$ equals $v_H$: \[\bigslant{G \sqcup H}{v_G\sim v_H}.\]

In general, computing $\mathcal{U}(I_{G\star H})$ can be difficult given $\mathcal{U}(I_G)$ and $\mathcal{U}(I_H)$ since new primitive closed even walks could be formed using odd cycles of $G$ and $H$ being linked through $v_G=v_H$. Furthermore, these odd cycles are not explicitly recorded in the list of primitive closed even walks, so we can't expect to compute the new list using the previous two lists alone. However, there is a special case where this operation works well. We start with an illustrative example. 

\begin{example}

Consider the graphs $G$ and $H$ pictured below.

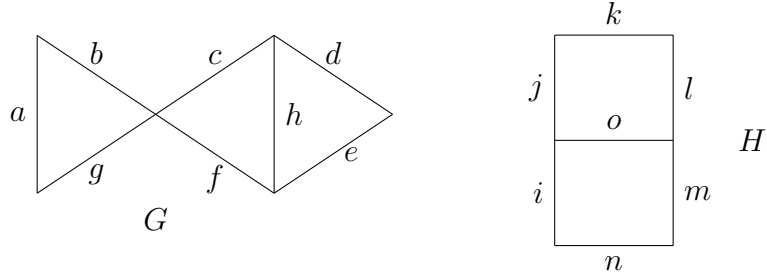
\begin{figure}[h]
\centering
\begin{minipage}{.5\textwidth}
  \centering
\begin{tikzpicture}[scale=0.35]
      % graph G
 \draw[dotted] (4.5,2) node{$G$};    
      \draw (0,3) -- (0,9)node[midway,  left] {$a$};
      \draw (4.5,6) -- (0,9) node[midway, above] {$b$};
      \draw (0,3) -- (4.5,6) node[midway, below] {$g$};
      \draw (4.5,6) -- (9,3) node[midway, below] {$f$};
      \draw (9,3) -- (9,9) node[midway, right] {$h$};
      \draw (4.5,6) -- (9,9) node[midway, above] {$c$};
      \draw (9,9) -- (13.5,6) node[midway, above] {$d$};
      \draw (13.5,6) -- (9,3) node[midway, right] {$e$};

       \end{tikzpicture}

\end{minipage}%
\begin{minipage}{.25\textwidth}
  \centering
\begin{tikzpicture}[scale=0.35]
\draw[dotted] (3,-1) node{$H$};   
\draw (0,3) -- (-4.5,3) node[midway, above] {$k$};
      \draw (0,3) -- (0, -1) node [midway, right] {$l$};
      \draw (-4.5,3) -- (-4.5, -1) node [midway, left] {$j$};
      \draw (-4.5, -1) -- (0, -1) node [midway, above] {$o$};
      \draw (-4.5, -1) -- (-4.5, -5) node [midway, left] {$i$};
      \draw (0, -1) -- (0, -5) node [midway, right] {$m$};
      \draw (-4.5, -5) -- (0, -5) node [midway, below] {$n$};

      \end{tikzpicture}
\end{minipage}
\end{figure} 

The set of primitive closed even walks for each can be directly computed as \[\mathcal{U}(I_G) = \{ce-df, acf-bgh, ac^2e-bdgh, adf^2-begh\}\] 
\[\mathcal{U}(I_H) = \{im-no, jl-ok, ikm-jln\}.\]

We can create a new graph $G \star H$ by gluing on $H$ at a vertex of $G$, say at the vertex incident to $a$ and $g$ in $G$, and $k$ and $\ell$ in $H$: 
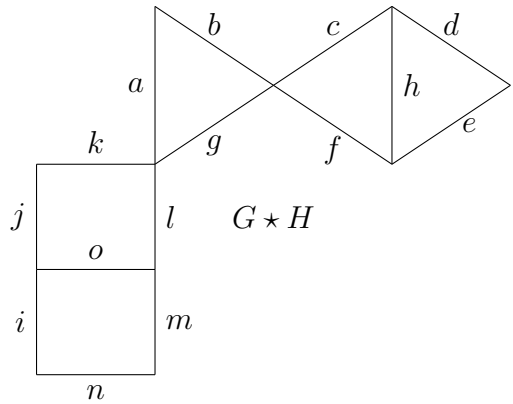
\begin{center}
    \begin{tikzpicture}[scale=0.35]
      % graph G
   \draw[dotted] (4.5,1) node{$G\star H$};      
      \draw (0,3) -- (0,9)node[midway,  left] {$a$};
      \draw (4.5,6) -- (0,9) node[midway, above] {$b$};
      \draw (0,3) -- (4.5,6) node[midway, below] {$g$};
      \draw (4.5,6) -- (9,3) node[midway, below] {$f$};
      \draw (9,3) -- (9,9) node[midway, right] {$h$};
      \draw (4.5,6) -- (9,9) node[midway, above] {$c$};
      \draw (9,9) -- (13.5,6) node[midway, above] {$d$};
      \draw (13.5,6) -- (9,3) node[midway, right] {$e$};
      \draw (0,3) -- (-4.5,3) node[midway, above] {$k$};
      \draw (0,3) -- (0, -1) node [midway, right] {$l$};
      \draw (-4.5,3) -- (-4.5, -1) node [midway, left] {$j$};
      \draw (-4.5, -1) -- (0, -1) node [midway, above] {$o$};
      \draw (-4.5, -1) -- (-4.5, -5) node [midway, left] {$i$};
      \draw (0, -1) -- (0, -5) node [midway, right] {$m$};
      \draw (-4.5, -5) -- (0, -5) node [midway, below] {$n$};

       \end{tikzpicture}
    
\end{center}
We can check that the set of primitive closed even walks for the resulting graph is the union of both lists \[\mathcal{U}(I_{G \star H}) = \{ce-df, acf-bgh, ac^2e-bdgh, adf^2-begh, im-no, jl-ok, ikm-jln\}.\] 
In fact, we would have arrived at the same result if we chose any other pair of vertices to identify. \hfill $\square$
\end{example}

In general, $\mathcal{U}(I_{G\star H})$ contains the union of $\mathcal{U}(I_G)$ and $\mathcal{U}(I_H)$. When $H$ is a bipartite graph however (i.e. contains no odd-length cycles), we get the reverse containment too. The next proposition is motivated by \cite[Section 2.0.3]{N}.

\begin{proposition}\label{prop: glue bipartite}
Let $G$ and $B$ be finite simple graphs such that $B$ is bipartite and $V(G)\cap V(B)=\emptyset$. Let $v_G\in V(G)$ and $v_B\in V(B)$, and form a new graph $G\star B$ by identifying $v_G$ and $v_B$. Then 
\[\mathcal{U}(I_{G\star B})=\mathcal{U}(I_G)\sqcup \mathcal{U}(I_B).\]

\end{proposition}
\begin{proof}
One direction is clear, since any primitive closed even walk of $G$ or $B$ must remain primitive in $G\star B$. Therefore $\mathcal{U}(I_G)\sqcup \mathcal{U}(I_B)\subseteq \mathcal{U}(I_{G\star B})$.

For the other direction, note that by \cite{TT} (and rephrased in \cite[Theorem 1.7]{DSNR}), a primitive closed even walk is either an even cycle, or contains at least two odd cycles. If $G\star B$ has a primitive closed even walk $\Gamma$ involving odd cycles, then these odd cycles must be in $G$ since $B$ is bipartite. 
If $\Gamma$ includes an edge of $B$, then the walk must pass through $v_B=v_G$ at least twice (in order to start and end in $G$). The edges between the first instance of $v_B$ and the second instance will define an even cycle $\Gamma_B$ of $B$, which is not possible by \cite[Lemma 2.2 (ii)]{GHKKPVT}.

Similarly, if there is some even cycle of $G\star B$ that is not contained in $G$ or $B$ exclusively, then we can write it as 
\[v_{i_1},e_{i_1},v_{i_2},e_{i_2},\ldots, e_{i_k},v_{i_{k+1}}\]

\noindent where all $e_i$ and $v_i$ are distinct except for $v_{i_1}=v_{i_{k+1}}$. If the cycle uses edges in $B$, then $v_B$ would appear twice in the list, unless $v_{i_1}=v_{i_{k+1}}=v_B$, which would mean that all of the edges are either entirely in $G$ or entirely in $B$, a contradiction.
\end{proof}

\subsubsection{Star contractions and subdivisions}

Next, we will consider a graph operation called a star contraction. Its use in the context of toric ideals of graphs was first introduced in \cite{N}. 

\begin{definition}\cite[Definition 3.4]{N}\label{def: star contraction}
Let $G$ be a graph with $v\in V(G)$, and $N_E(v)$ be the list of edges in $E(G)$ which are incident to $v$. The \emph{star contraction} of $G$ at $v$ is the graph $G_v$ formed by performing an edge contraction on all of the edges in $N_E(v)$ simultaneously. That is, $G_v$ is constructed by first deleting all edges in $N_E(v)$, and then identify all vertices in the neighborhood of $v$.  
\end{definition} 

\begin{example}
\cite[Example 3.0.6]{N} Consider the star contraction of the graph $G$ below along the vertex $v$ incident to $e$ and $f$. The list of primitive closed even walks of $G$ and $G_{v}$ have also been listed. Notice that we can get the list of elements in $\mathcal{U}(I_{G_v})$ from $\mathcal{U}(I_G)$ by setting $e=f=1$.
\newpage 

\begin{figure}[ht]
\centering
\begin{minipage}{.5\textwidth}
  \centering
\begin{tikzpicture}[scale=1]

\draw (-2,2) -- (0,2) node[midway, above] {$a$};
      \draw (0,2) -- (2,2) node [midway, above] {$b$};
      \draw (2,2) -- (2,0) node [midway, right] {$c$};
      \draw (2,0) -- (0,0) node [midway, below] {$d$};
      \draw (0,0) -- (-2,0) node [midway, below] {$e$};
      \draw (-2,0) -- (-2,2) node [midway, left] {$f$};
      \draw (0,2) -- (0, 0) node [midway, right] {$g$};
    \end{tikzpicture}

\end{minipage}$\hspace{-10mm}\longrightarrow$
\begin{minipage}{.3\textwidth}
  \centering
\begin{tikzpicture}[scale=1]

\draw (0,0)  to[out=-235, in=235](0,2) node [midway,above, left] {$a$};
      \draw (0,2) -- (2,2) node [midway, above] {$b$};
      \draw (2,2) -- (2,0) node [midway, right] {$c$};
      \draw (2,0) -- (0,0) node [midway, below] {$d$};
      \draw (0,2) -- (0, 0) node [midway, right] {$g$};
      
    \end{tikzpicture}
\end{minipage}

\[\hspace{1.0cm}\langle ace-bdf, ae-fg, bd-cg\rangle  \hspace{0.8cm} \longrightarrow \hspace{0.8cm} \langle ac-bd, a-g, bd-cg\rangle \]
\end{figure} \hfill $\square$

\end{example}

To simplify notation, we will define the ring homomorphism \[\pi_v: \mathbb{K}[E(G)] \rightarrow \mathbb{K}[E(G)\setminus N_E(v) ] \] on generators by $e\mapsto 1$ if $e\in N_E(v)$, and $e\mapsto e$ otherwise. To avoid any issues with defining primitive walks for multigraphs (like in the previous example), we will restrict to the case when the star contraction results in a simple graph.

\begin{lemma}\label{lem: star substitution}\cite[Theorem 3.10]{N}
Let $G$ be a finite simple graph. Suppose that $v\in V(G)$ is such that $G_v$ is a simple graph. Then

\[ \mathcal{U}(I_{G_v})\subseteq \pi_v(\mathcal{U}(I_G))\]
\end{lemma}

These results allow us to produce new graphs through star contractions while still having control over the enumeration of primitive closed even walks. Note that even though the containment in Lemma \ref{lem: star substitution} is generally proper, the set $\pi_v(\mathcal{U}(I_G))$ still defines a universal Gr\"obner basis of $I_{G_v}$ (although not a reduced basis).

This operation can also be undone to produce larger graphs, a process called a \emph{star subdivision}, generally discussed in \cite{BC} for toric ideals of graphs. We will show that in the special case when the subdivision is done along a vertex of degree 2, the list of primitive closed even walks has an explicit description. To do this, consider a graph with the following structure:

 \begin{figure}[!ht]
    \centering \hspace{-10mm}
    \begin{minipage}{.5\textwidth}
  \centering
    \begin{tikzpicture}[scale=0.55]
      % graph G
      \draw[dotted] (6.3,0) circle (3cm) node{$G$};
      \draw[thin,dashed] (4.5,3) -- (6,3);
      \draw[thin,dashed] (4.5,-3) -- (6,-3);
      \draw (4.5,-3) -- (3,-2.) node[midway,below ] {\footnotesize{$x$}};
       \draw (2,-0.5) -- (3,-2) node[midway,below, left ] {\footnotesize{$y'$}};
      \draw (2,-0.5) -- (3,1.5) node[midway, above,left] {\footnotesize{$x'$}};
       \draw (3,1.5) -- (4.5,3) node[midway,above ] {\footnotesize{$y$}};

      % draw nodes
      \fill[fill=white,draw=black] (2,-0.5)  circle (.1) node[right] {\footnotesize{$v$}}; 
      \fill[fill=white,draw=black] (3,1.5) circle (.1);
      \fill[fill=white,draw=black] (3,-2) circle (.1);

    \end{tikzpicture}
    \end{minipage}
    $\hspace{-10mm}\longrightarrow$
\begin{minipage}{.3\textwidth}
  \centering

      \begin{tikzpicture}[scale=0.55]
      % graph G
      \draw[dotted] (6.3,0) circle (3cm) node{$G_v$};
      \draw[thin,dashed] (4.5,3) -- (6,3);
      \draw[thin,dashed] (4.5,-3) -- (6,-3);
      \draw (4.5,-3) -- (2,-0) node[midway,below ] {\footnotesize{$x$}};
       \draw (2,0) -- (4.5,3) node[midway,above ] {\footnotesize{$y$}};

      % draw nodes
      \fill[fill=white,draw=black] (2,0)  circle (.1) node[right] {\footnotesize{$v$}};

    \end{tikzpicture}
    \end{minipage}
 \end{figure}

\noindent An important feature of such a graph is that the star contraction along the vertex $v$ incident to $x'$ and $y'$ results in another degree 2 vertex (which we also call $v$ by an abuse of notation). In this case, we will say that $G$ is the (unique) star subdivision of $G_v$ along $v$. More generally, there are usually multiple star subdivisions of a graph (see \cite[Definition 3.0.3]{BC}) if the degree of $v$ is greater than two.

To demonstrate the effect on the list of primitive closed even walks after the star subdivision, consider the following map on polynomial rings, \[\psi_v: \mathbb{K}[\mathbf{e},x,y] \rightarrow \mathbb{K}[\mathbf{e},x,y,x',y']\] defined by $x\mapsto xx'$, $y\mapsto yy'$, and $f\mapsto f$ for $f\in\mathbf{e}=E(G)\setminus\{x,x',y,y'\}$. Notice that if $m_1x-m_2y$ is a closed even walk of $G_v$ (where $m_1,m_2$ are monomials with support in $\mathbf{e}$), then $\psi_v(m_1x-m_2y) = m_1xx'-m_2yy'$ is a closed even walk of $G$. The next result shows that the same is true for primitive walks.

\begin{proposition}\label{prop: star degree 2}
Let $G$ be a finite simple graph and suppose that $v\in V(G)$ has degree 2 in $G$ and degree 2 in the star contraction $G_v$. Let the edges incident to $v$ be labeled as above. Then \[\mathcal{U}(I_{G_v}) = \pi_v(\mathcal{U}(I_G))\]
\noindent and \[\mathcal{U}(I_{G}) = \psi_v(\mathcal{U}(I_{G_v}) ). \]
\end{proposition}
\begin{proof}
First note that for $\Gamma\in I_G$ and $\gamma\in I_{G_v}$: \[\psi_v(\pi_v(\Gamma)) = \Gamma  \hspace{4mm}\text{ and }\hspace{4mm} \pi_v(\psi_v(\gamma))=\gamma, \hspace{4mm}\] so it suffices to prove the first equality to show that the second claim is also true.

 By the structure of primitive closed even walks (see \cite{TT}), any primitive closed even walk of $G$ that passes through $v$ must also pass through the edges $x,y,x'$ and $y'$. Furthermore, it would either pass through all $4$ edges exactly once or twice. Therefore, all binomials in $\mathcal{U}(I_G)$ which correspond to a primitive walk passing through $v$ must be of the form:
\[m_1xx'-m_2yy' \hspace{4mm} \text{ or } \hspace{4mm} m_1(xx')^2-m_2(yy')^2 \]
\noindent where $m_1$ and $m_2$ are monomials with support in  $\mathbf{e}=E(G)\setminus\{x,x',y,y'\}$.

Let $\Gamma=m_1xx'-m_2yy'$ be a primitive closed even walk of $G$. Assume that $\pi_v(\Gamma)=m_1x-m_2y$ is not primitive in $G_v$. Then there would be some other binomial $\gamma= m_3-m_4\in I_{G_v}$ such that $m_3|m_1x$ and $m_4|m_2y$. If the support of $m_3$ and $m_4$ does not include $x$ or $y$, then $\gamma$ is unaffected by the star subdivision of $G_v$ and corresponds to a closed even walk of both $G$ and $G_v$, contradicting that $\Gamma$ is primitive. 

The only other case is when $x|m_3$ and $y|m_4$ (since $\gamma$ passing through $v$ must also pass through both $x$ and $y$). In this case, $\psi_v(\gamma) = m_3x'-m_4y'$ defines a closed even walk of $G$ such that $m_3x'|m_1xx'$ and $m_4y'|m_2yy'$, also contradicting the fact that $\Gamma$ is primitive. The case $\Gamma=m_1(xx')^2-m_2(yy')^2$ is similar. Since walks that do not pass through $v$ are unaffected by the star contraction, we have shown that $\pi_v(\mathcal{U}(I_G)) \subseteq \mathcal{U}(I_{G_v})$. Together with Lemma \ref{lem: star substitution}, we have shown that $\mathcal{U}(I_{G}) = \psi_v(\mathcal{U}(I_{G_v}) )$, as required.
\end{proof}

\subsubsection{Gluing even cycles} \label{subsubsec: cycle gluing}

Finally, we can obtain new graphs for which we can recursively generate the list of primitive closed even walks using cycle gluing. We can do this similarly to the vertex gluing defined earlier, except that we identify two edges instead of two vertices. More specifically, given disjoint graphs $G$ and $H$, and edges $e_G\in E(G)$ and $e_H\in E(H)$, we can produce a new graph of the form \[\bigslant{G \sqcup H}{e_G\sim e_H}\] which we denote by $G*_{e_G,e_H}H$ (or simply $G*H$ when $e_G$ and $e_H$ have already been specified). The use of cycle gluing in the context of toric ideals of graphs and geometric vertex decomposition was introduced in \cite[Theorem 3.11]{CDSRVT}. In the proof of that result, the structure of $\mathcal{U}(I_{G*H})$ was described, which we demonstrate in the next example.

\begin{example}
Consider the graphs $G$ and $H$ 

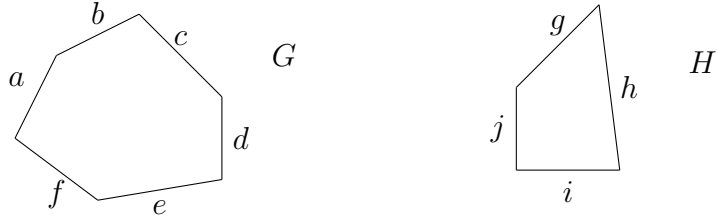
\begin{figure}[h]
\centering
\begin{minipage}{.5\textwidth}
  \centering
\begin{tikzpicture}[scale=0.55]
\draw[dotted] (4.5,2) node{$G$};   
     \draw (-2,0) -- (-1, 2) node [midway, above left] {$a$};
     \draw (-1,2) -- (1 ,3) node [midway, above] {$b$};
     \draw (1,3) -- (3, 1) node [midway, above] {$c$};
     \draw (3,1) -- (3,-1) node [midway, right] {$d$};
     \draw (0, -1.5) -- (3, -1) node [midway, below] {$e$};
     \draw (-2, 0) -- (0, -1.5) node [midway, below] {$f$};
        \end{tikzpicture}

\end{minipage}%
\begin{minipage}{.25\textwidth}
  \centering
\begin{tikzpicture}[scale=0.55]
\draw[dotted] (7.5,1.6) node{$H$};   
 \draw (3,1) -- (5, 3) node [midway, above] {$g$};
     \draw (5,3) -- (5.5, -1) node [midway, right] {$h$};
     \draw (5.5,-1) -- (3,-1) node [midway, below] {$i$};
     \draw (3,1) -- (3,-1) node [midway, left] {$j$};

      \end{tikzpicture}
\end{minipage}
\end{figure} 

\noindent where $\mathcal{U}(I_G) = \{ ace-bdf \}$ and $\mathcal{U}(I_H) = \{ hj-ig \}$. We can define a new graph $G * H$ by gluing along two edges, say $d$ and $j$ (which we call $k$ after the identification).

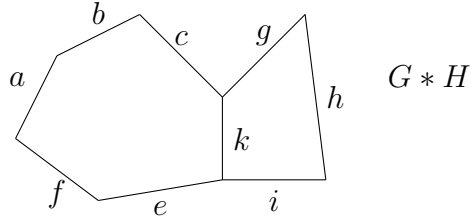
\begin{center}
\begin{tikzpicture}[scale=0.55]
        \draw[dotted] (8,1.5) node{$G*H$};   
      \draw (-2,0) -- (-1, 2) node [midway, above left] {$a$};
     \draw (-1,2) -- (1 ,3) node [midway, above] {$b$};
     \draw (1,3) -- (3, 1) node [midway, above] {$c$};
     \draw (3,1) -- (5, 3) node [midway, above] {$g$};
     \draw (5,3) -- (5.5, -1) node [midway, right] {$h$};
     \draw (5.5,-1) -- (3,-1) node [midway, below] {$i$};
     \draw (3,1) -- (3,-1) node [midway, right] {$k$};
     \draw (0, -1.5) -- (3, -1) node [midway, below] {$e$};
     \draw (-2, 0) -- (0, -1.5) node [midway, below] {$f$};    
        \end{tikzpicture}
    \end{center}
The list of primitive closed even walks for $G*H$ becomes \[\mathcal{U}(G * H) =\langle ace-bfk, hk-ig, aceh-bfgi \rangle \]
where one additional walk is produced by extending $ace-bdf$ to bypass $d=k$ and transverse the even cycle instead.\hfill $\square$
\end{example}

The ``extended'' walks from the example are formed by taking any walk through the edge used for gluing and extending the walk to traverse the even cycle. We make this more precise as follows. Let $G$ be a finite simple graph with $e_G\in E(G)$, and $C_{2k}$ be some disjoint cycle with $e_C\in E(C_{2k})$. We will glue $G$ to $C_{2k}$ along $e_G$ and $e_C$ to produce a new graph $G*C_{2k}$. Suppose that $\gamma=u_1e_G^\ell-v_1 \in \mathcal{U}(I_G)$ where $\ell=1,2$, and $u_2e_C-v_2\in \mathcal{U}(I_{C_{2k}})$ is the binomial defined by the cycle $C_{2k}$. Here $u_1,u_2,v_1,v_2$ are monomials with support not including $e_G$ or $e_C$. If $e_G$ is glued to $e_C$ and relabeled as $e$, then the extension of $\gamma$, denoted by $\bar{\gamma}$, is the binomial $u_1v_2 - v_1u_2$ if $\ell=1$ and $u_1v_2^2 - v_1u_2^2$ if $\ell=2$. It is not difficult to see that both of these define primitive closed even walks contained in $\mathcal{U}(I_{G * C_{2k}})$.

\begin{proposition}\label{prop: glue even cycle}
Let $G$ be a finite simple graph and $C_{2k}$ be a disjoint cycle of length $2k$, $k>1$. Let $e_G\in E(G)$ and $e_C\in E(C_{2k})$, and form a new graph $G* C_{2k}$ by identifying $e_G$ and $e_C$ as the edge $e$. If $\Gamma\in \mathcal{U}(I_{G * C_{2k}})$, then either:

\begin{itemize}
    \item $\Gamma\in\mathcal{U}(I_G)$
    \item $\Gamma$ is the binomial defining $C_{2k}$
    \item $\Gamma = \bar{\gamma}$ for some $\gamma\in \mathcal{U}(I_G)$ which passes through $e_G$
\end{itemize}
\end{proposition}
\begin{proof}
We will abuse notation and write $G$ and $C_{2k}$ for the subgraphs of $G*C_{2k}$ used to construct the gluing. Suppose that some edge of $C_{2k}$ appears in a primitive closed even walk $\Gamma$ of $G*C_{2k}$. Then $\Gamma$ is either an even cycle or a primitive walk containing at least two odd cycles \cite[Theorem 1.7]{DSNR}. In the first case, we follow the argument of the proof of Theorem 3.11 in \cite{CDSRVT} to conclude that the even cycle is either $C_{2k}$ itself, is an even cycle of $G$, or has the form $\bar{\gamma}$ where $\gamma$ is an even cycle of $G$ which passes through $e$.

If $\Gamma$ includes at least two odd cycles and is not exclusively in $G$, then it must be of the form $\bar{\gamma}$ for some primitive walk $\gamma$ of $G$ which passes through $e$. Indeed, if there is a walk $\Gamma$ that includes the edges of $C_{2k}$, then it must pass through the endpoints of $e$, so let $\gamma$ be the walk where the sequence of edges of $C_{2k}\setminus e$ in $\Gamma$ are replaced by $e$. Let $u_2e-v_2$ be the binomial of $C_{2k}$ in $G*C_{2k}$, where $u_2,v_2$ are monomials with support in $E(C_{2k}\setminus e)$. There are now two cases to consider:

\noindent\underline{Case 1:} If $e$ appears exactly once in the walk $\gamma$ so that $\gamma=u_1e-v_1$ (where $u_1,v_1$ are monomials with support in $E(G\setminus e)$), then the proof of Theorem 3.11 in \cite{CDSRVT} shows that $\bar{\gamma} = u_1v_2 - v_1u_2$. To show that it is primitive, observe that any $u_3-v_3\in I_{G*C_{2k}}$ which doesn't pass through $e$ will either use all variables in $E(C_{2k}\setminus e)$, or will not use any of the variables of the cycle. Assume that $u_3|u_1v_2$ and $v_3|v_1u_2$. Then we can write $u_3 = c_3g_3$ and $v_3=c_4g_4$ where $c_3|v_2$, $c_4|u_2$, $g_3|u_1$ and $g_4|v_1$. Then either $\gamma$ is not primitive because of $g_3e-g_4$, or the binomial for  $C_{2k}$ is not primitive because of $c_4e-c_3$, which is a contradiction.

\noindent\underline{Case 2:} If $e$ appears twice in $\gamma$, then we can write $\gamma=u_1e^2-v_1$. As above, we can show that $\bar{\gamma}=u_1v_2^2 - v_1u_2^2$, by tracing out $\gamma$ in the following way. Let $e=\{a,b\}$. Start at vertex $a$, and trace out the portion of $\gamma$ that start at $a$, stays in $G\setminus e$, and returns to vertex $a$. Then cross through all edges in $C_{2k}\setminus e$ to get to vertex $b$. Then trace out the portion of $\gamma$ that starts at vertex $b$ and stays in $G\setminus e$, returning to vertex $b$. Finally, cross the edges of $C_{2k}\setminus e$ again to get back to $a$. Note that the intermediate vertices in a primitive walk can only be visited twice (since every cut vertex only belongs to two blocks by \cite[Theorem 2.2]{TT}). We can show that $\bar{\gamma}$ is primitive using a similar argument as above. 
\end{proof}

\subsubsection{Main Theorem}

Using the previously mentioned operations, we are now ready to show that arbitrarily large universal Gr\"obner bases can be produced in polynomial time. Starting with a small graph where $\mathcal{U}(I_G)$ can be computed directly, and through random applications of each operation, a sufficiently large (and asymmetric) graph $H$ with computable $\mathcal{U}(I_H)$ can be constructed to secure the system $\mathcal{P}$. By asymmetric, we mean that repetitive iterations of the same operation should be avoided (such as simply gluing on a 4-cycle successively).

\begin{theorem}\label{thm: construct U_I}
Let $G$ be a finite simple graph such that $\mathcal{U}(I_G)$ is known. Then by using one of the following operations

\begin{enumerate}
    \item Gluing a disjoint bipartite graph to $G$ along some $v\in V(G)$ (as in Section  \ref{subsubsec: vertex gluing})
    \item Gluing a disjoint even cycle to $G$ along some $e\in E(G)$ (as in Section \ref{subsubsec: cycle gluing})
    \item Star subdividing along a degree two vertex of $G$ (as in Proposition \ref{prop: star degree 2})
    \item Performing a star contraction along a vertex $v$ such that $G_v$ is a simple graph (as in Definition \ref{def: star contraction})
\end{enumerate}
we can produce a graph $G'$ such that the number of operations to compute $\mathcal{U}(I_{G'})$ is linear in $N=|\mathcal{U}(I_G)|$. Furthermore, by using any combination of $k$  operations $(1)$ to $(4)$, and choosing $k$ sufficiently large, we can produce a graph $H$ such that $|\mathcal{U}(I_H)|$ is as large as desired, with computational complexity $O(N^k)$.
\end{theorem}
\begin{proof}
For the first three operations, the explicit method in which $\mathcal{U}(I_{G'})$ is obtained from the $\mathcal{U}(I_G)$ is described in Propositions \ref{prop: glue bipartite} and \ref{prop: glue even cycle} and also Proposition \ref{prop: star degree 2}. Here we would produce a larger list $\mathcal{U}(I_{G'})$ given $\mathcal{U}(I_G)$ for the first two operations, while the third operation would maintain the cardinality of the sets but increase the degree. 

The fourth operation maintains the same cardinality, although the new list of closed even walks may not all be primitive (this is still okay in the context of Gr\"obner bases since we are simply adding generators which may be unnecessary for the Gr\"obner computation). 

If $|\mathcal{U}(I_G)| = N$, then the first operation simply merges two sets, which is done in linear time. The second operation requires at most $N+1$ new elements to be added to the list of primitive closed even walks (one instance of $C_{2k}$, and at most one $\bar{\gamma}$ computation for each $\gamma\in \mathcal{U}(I_G)$), so has complexity $O(N)$. The third operation increases the degree of at most $N$ walks, which again has complexity $O(N)$. Finally, the star contraction requires a substitution of at most $N$ polynomials, which is again $O(N)$. Iterating these operations would result in the product of the complexity bounds, proving the $O(N^k)$ bound. 
\end{proof}

\section{Conclusions and Alternate Protocols} \label{sec: conclusion}

We conclude with some brief observations about the use of universal Gr\"obner bases for securing data. The protocol $\mathcal{P}$ is just one possible vision of how universal Gr\"obner bases could be used in cryptography. We hope that this article will spur interest in other possible uses of the construction of $\mathcal{U}_I$ proposed in Section \ref{sec: initial}, especially by those better versed with the practical issues concerning cryptographic implementations. 

We offer several remarks on alternate approaches:

\begin{itemize}
    \item Choosing $\mathbb{K}$ to be a finite field would increase the difficulty of the Gr\"obner computations and would likely improve the security of $\mathcal{P}$.
    \item When $\tau=\emptyset$, one shortfall of the system is the amount of time that Party $A$ needs to decrypt the message. This may make the $\tau=\emptyset$ protocol useful for blockchain applications where rewards are used to incentivize the completion of brute-force verifications. 
    \item Symmetric Diffie-Hellman type initializations of $\mathcal{P}$ may be possible by Party $A$ providing a common monomial ideal, followed by $A$ and $B$ each choosing their own initial ideals and combining it with this common ideal. Sending such ``combined" ideals (using unions, intersections, etc.) may reveal too much information about degree bounds of generators in the choices of $A$ and $B$. Masking the choices using hash functions would yield a similar security to the $\tau\neq \emptyset$ case.
    \item Party $B$ only sending partial information about $K_B$ would reduce the number of keys that $A$ needs to check, offering a middle ground between the $\tau=\emptyset$ and $\tau\neq \emptyset$ initializations of $\mathcal{P}$. Choosing $\mathcal{E}$ to be a symmetric encryption algorithm would also reduce $A$'s decryption time, since such schemes are usually less computationally intensive compared to their asymmetric counterparts.
\end{itemize}

As a final note, universal Gr\"obner bases for toric ideals have been better studied and are generally faster to compute. Using some  $\mathcal{U}_I$ associated to the toric ideal of a graph may introduce a weakness to the system if chosen poorly. Generally, the complexity of computing universal Gr\"obner bases for a toric ideal of a graph still remains exponential in the number of edges \cite[Section 4]{ST}. 

An alternate approach is to build a large enough $\mathcal{U}_I$ using the techniques in Section \ref{sec: initial}, and then add one (carefully selected) non-toric generator to the list, followed by a recomputation of  a universal Gr\"obner basis for the new list. If an attacker does not know the graph $G$, then the toric universal Gr\"obner basis algorithms from \cite{CLS,S} would be difficult to implement. Furthermore, even if $G$ were known, choosing it large enough would make those computations difficult.

\end{document}